\documentclass[conference]{IEEEtran}

\IEEEoverridecommandlockouts

\usepackage[utf8]{inputenc} 
\usepackage[T1]{fontenc}
\usepackage{url}
\usepackage{comment}
\usepackage{cite}
\usepackage[cmex10]{amsmath} % Use the [cmex10] option to ensure complicance
\usepackage{amssymb}
\usepackage{xcolor}
\usepackage{graphicx}
\usepackage{dirtytalk}

\newcommand{\F}{\mathbb{F}}
\newcommand{\Z}{\mathbb{Z}}
\newcommand{\E}{\mathbb{E}}
\newcommand{\dH}{\mathcal{D}}
\newcommand{\I}{\mathcal{I}}
\newcommand{\cl}{\operatorname{cl}}

\newcommand{\code}{\mathcal{C}}

\newcommand{\ttau}{\widetilde{\tau}}
\newcommand{\aalpha}{\widetilde{\alpha}}
\newcommand{\duale}[1]{#1^{\perp}}

\renewcommand{\le}{\leqslant}
\renewcommand{\leq}{\leqslant}
\renewcommand{\ge}{\geqslant}
\renewcommand{\geq}{\geqslant}

\newcommand{\Cref}[1]{Co\-ro\-lla\-ry\,\ref{#1}}



\outer\def\proclaim #1. #2\par{\medbreak
 \noindent{\bf#1.\enspace}{\sl#2\par}%
 \ifdim\lastskip<\medskipamount \removelastskip\penalty55\medskip\fi}

\newtheorem{theorem}{Theorem}
\newtheorem{proposition}{Proposition}
\newtheorem{definition}{Definition}

\newtheorem{lemma}{Lemma}
\newtheorem{corollary}{Corollary}

\newtheorem{remark}{Remark}
\newtheorem{conjecture}{Conjecture}

\makeatletter
\def\@begintheorem#1#2{%
  \trivlist
  \item[\hskip\labelsep{\bfseries #1\ #2.}]%
}
\def\@opargbegintheorem#1#2#3{%
  \trivlist
  \item[\hskip\labelsep{\bfseries #1\ #2\ (#3).}]%
}
\makeatother

\begin{document}

\title{A Matroidal Framework\\ for Random Access in Linear Codes%
\thanks{The work of D. Bar-Lev was supported in part by the Swiss National Science Foundation under Grant 212865 and by Schmidt Sciences. H. H. L\'opez was partially supported by the NSF grants DMS-2401558 and DMS-2502705. G. L. Matthews is supported in part by NSF DMS-2502705 and the Commonwealth Cyber Initiative. F. Salizzoni is supported by the Swiss National Science Foundation under the grant P5R5-2\_239500}}

\author{%
\IEEEauthorblockN{%
Daniella Bar-Lev\IEEEauthorrefmark{1}\IEEEauthorrefmark{2},
Hiram H. L\'opez\IEEEauthorrefmark{3},
Gretchen L. Matthews\IEEEauthorrefmark{3},
and Flavio Salizzoni\IEEEauthorrefmark{1}}
\IEEEauthorblockA{%
\IEEEauthorrefmark{1}\textit{Department of Mathematics,
Universit\"at Z\"urich}, Z\"urich, Switzerland\\
\IEEEauthorrefmark{2}\textit{Department of Computer Science,
Technion---Israel Institute of Technology}, Haifa 3200003, Israel\\
\IEEEauthorrefmark{3}\textit{Department of Mathematics,
Virginia Tech}, Blacksburg, VA, USA\\
daniellalev@technion.ac.il,\quad
\{hhlopez,gmatthews\}@vt.edu,\quad
flavio.salizzoni@math.uzh.ch}}
\maketitle

\begin{abstract}
We study random access in linear codes when codeword symbols are sampled independently and uniformly with replacement. A code is recovery balanced if the expected number of samples required to recover a codeword coordinate is the same for every coordinate. We express the recovery expectations in terms of the associated matroid and derive a coordinatewise identity relating these values for a code and its dual. This proves that a linear code is recovery balanced if and only if its dual is recovery balanced, establishing a conjecture of Gruica et al. We further define the recovery balance for arbitrary matroids and show that asymptotically almost every sparse paving matroid is not recovery balanced.
\end{abstract}


\section{Introduction}
The coded coupon collector problem studies the recovery of information from randomly sampled encoded symbols. An information vector is encoded into a codeword, whose symbols are sampled independently and uniformly with replacement until the collected symbols suffice to recover the requested information. The objective is to minimize the expected number of samples required for recovery. Introduced in~\cite{barlev2025cover} with motivation from DNA-based data storage systems~\cite{church2012next,goldman2013towards,grass2015robust,erlich2017dna,organick2018random,bar2025scalable}, this problem has since been studied under several recovery objectives.

In the full-recovery setting, the entire information vector must be reconstructed. Work in this direction includes exact recovery-time formulas, optimality results for specific code families, and the analysis of random linear codes~\cite{barlev2025cover,bertuzzo2026duality,gross2026optimal}. Other formulations require recovering a requested collection of files~\cite{abraham2024covering} or information symbols~\cite{gruica2026generalized}, and study the worst-case or average expected recovery time over requests of a given size. A complementary direction fixes a partition into files and investigates the tradeoff between their individual expected retrieval times~\cite{barlev26coded,vlachos2026hyperbolic}. Partial reconstruction, where only a fraction of the information needs to be recovered, has also been studied using Luby Transform (LT) codes~\cite{akkaya2026dupleth}. 

Alongside these different recovery objectives, several works extend the sampling and storage models. These include recovery under noisy or nonuniform sampling~\cite{hanna2026reliability,cao2026modeling,abraham2026coverage}, reconstruction of a failed storage container from surviving containers~\cite{levy2026expected}, and full recovery through parallel sampling from multiple containers~\cite{kong2026coverage}. Related generalizations of the classical coupon collector problem allow each draw to reveal a subset of items, including contiguous windows~\cite{grunbaum2026general}, as well as the reconstruction of composite symbols and combinations of motifs from repeated observations~\cite{preuss2024sequencing,cohen2025optimizing,sokolovskii2024coding}.

In this work, we consider the uniform noiseless model and focus on random access, where only a single information symbol is requested~\cite{gruica2025combinatorial,gruica2026geometry,boruchovsky2026making,bodur2026random,wang2026random}. For linear codes, recovery is possible once the requested symbol can be expressed as a linear combination of the sampled symbols. The expected recovery time therefore depends on the linear dependencies among the columns of the generator matrix.

To study this dependence, Gruica et al.~\cite{gruica2025combinatorial} considered the expected recovery times of all codeword coordinates and introduced \emph{recovery balanced codes}, for which these expectations are equal. They proved that the average expectation over the coordinates of an $[n,k]_q$ code is $k$, and hence every coordinate of a recovery balanced code has expected recovery time $k$. Under a systematic encoding, this also holds for every information symbol and equals the uncoded recovery time. Thus, improving on this baseline requires a code that is not recovery balanced. The same work conjectured that recovery balance is preserved under code duality.

In this paper, we observe that being recovery balanced is a matroidal property in the sense that, to determine whether a code is recovery balanced, it suffices to examine the associated matroid. We therefore introduce the notion of a recovery-balanced matroid and show that this property is preserved under duality if the matroid has neither loops nor coloops. In particular, it follows that a linear code is recovery balanced if and only if its dual is recovery balanced, establishing the conjecture over every finite field.

We then determine how common the recovery balanced property was among matroids. Using the correspondence with independent sets of Johnson graphs, we prove that, for every rank $1\le k\le n-1$, the proportion of recovery balanced sparse paving matroids on $n$ elements is at most $4/n$ for all sufficiently large $n$. Consequently, asymptotically almost every sparse paving matroid is not recovery balanced.

The rest of the paper is organized as follows. Section~\ref{sec:def} introduces the notation and basic definitions. Section~\ref{sec:reform} expresses coordinate recovery expectations as matroid invariants. Section~\ref{sec:main} derives the duality identity and proves that recovery balance is preserved under duality. Finally, Section~\ref{sec:sparse} characterizes recovery balance for sparse paving matroids and shows that asymptotically almost every sparse paving matroid is not recovery balanced.

\section{Notation and Definitions}\label{sec:def}

Throughout, $q$ is a prime power, $\F_q$ is the finite field with $q$ elements, and $[n]=\{1,\ldots,n\}$ for a positive integer $n$. We write $H_n=\sum_{j=1}^n 1/j$ for the $n$-th harmonic number, with
$H_0=0$, and $\langle\cdot\rangle$ for the span over $\F_q$.

\subsection{Linear Codes and Recovery Balance}

A \emph{linear $[n,k]_q$ code} is a $k$-dimensional subspace $\code\subseteq\F_q^n$. A \emph{generator matrix} of $\code$ is a rank-$k$ matrix $G\in\F_q^{k\times n}$ whose rows span $\code$. We denote its columns by $g_1,\ldots,g_n$. The \emph{dual code} is
the $[n,n-k]_q$ code
\[
  \duale{\code}
  =\{x\in\F_q^n:x\cdot c=0\text{ for all }c\in\code\}.
\]
A generator matrix of $\duale{\code}$ is a \emph{parity-check matrix}
of $\code$.
For background on linear codes and duality, see e.g.~\cite{macwilliams1977theory,roth2006introduction}.

We consider the random-access model of \cite{gruica2025combinatorial}, in which column indices of $G$ are drawn independently and uniformly from $[n]$, with replacement. For $i\in[n]$, let $\ttau_i(G)$ be the number of draws required until $g_i$ belongs to the span of the drawn columns. In particular,
$\ttau_i(G)=0$ if $g_i=0$.
This random variable was introduced in~\cite{gruica2025combinatorial}, where it was shown that its expectation depends
only on $\code$, so we write $\E[\ttau_i(\code)]$.

\begin{definition}[{\cite[Definition~2]{gruica2025combinatorial}}]
\label{def:rb}
A code $\code$ is \emph{recovery balanced} if
$\E[\ttau_1(\code)]=\cdots=\E[\ttau_n(\code)]$.
\end{definition}

For every $[n,k]_q$ code $\code$,
\cite[Theorem~1]{gruica2025combinatorial} gives
\begin{equation}\label{eq:average}
  \sum_{i=1}^n\E[\ttau_i(\code)]=kn.
\end{equation}
Consequently, if $\code$ is recovery balanced, then
$\E[\ttau_i(\code)]=k$ for every $i\in[n]$.
It was conjectured in~\cite{gruica2025combinatorial} that recovery balance is preserved under duality.
\begin{conjecture}[{\cite[Conjecture~1]{gruica2025combinatorial}}]
\label{conj:main}
A code $\code$ is recovery balanced if and only if its dual
$\duale{\code}$ is recovery balanced.
\end{conjecture}

Conjecture~\ref{conj:main} was proved for codes with a transitive permutation automorphism group. 
In Section~\ref{sec:main} we prove the conjecture for every linear $[n,k]_q$ code. As in~\cite{gruica2025combinatorial}, for $i\in[n]$ and an integer $0\le s\le n$, let
\[
  \aalpha_i(s)
  :=\bigl|\{S\subseteq[n]:|S|=s,\,
       g_i\in\langle g_j:j\in S\rangle\}\bigr|.
\]
The argument of \cite[Lemma~1]{gruica2025combinatorial}, applied to
recovery of $g_i$, gives
\begin{equation}\label{eq:expectation}
  \E[\ttau_i(\code)]
  =nH_n-\sum_{s=1}^{n-1}\frac{\aalpha_i(s)}{\binom{n-1}{s}},
\end{equation}
provided $g_i\ne0$, or equivalently $\aalpha_i(0)=0$.
Thus, $\aalpha_i(s)$ and the recovery expectations are determined
by the linear dependencies among the columns of $G$.
We express these dependencies using matroids, which are defined in the next section.

\subsection{Matroids and Duality}\label{sec:matroids}

A matroid is a combinatorial object that captures the notion of independence. Here, we recall the rank-based definition of a matroid and express independent sets, closure, and circuits in terms of the rank function. See~\cite{oxley2006matroid} for further background.

\begin{definition}
\label{def:matroid}
A \emph{matroid} is a pair $M=(E,r)$, where $E$ is a finite
\emph{ground set} and the \emph{rank function}
$r:2^E\to\Z_{\ge0}$ satisfies, for all $A,B\subseteq E$,
\begin{enumerate}
  \item[\textnormal{(R1)}] $0\le r(A)\le|A|$;
  \item[\textnormal{(R2)}] $r(A)\le r(B)$ if $A\subseteq B$;
  \item[\textnormal{(R3)}]
  $r(A\cup B)+r(A\cap B)\le r(A)+r(B)$.
\end{enumerate}
\end{definition}

Let $\code$ be an $[n,k]_q$ code with a generator matrix $G$. The matroid associated with~$\code$ 
is~$M[G]=([n],r)$, where
\[
  r(A)=\dim_{\F_q}\langle g_j:j\in A\rangle,
  \qquad A\subseteq[n].
\]
If $G$ generates $\code$, we write $M(\code)=M[G]$ and call it
the \emph{matroid of the code}. This is independent of the choice
of $G$, since invertible row operations preserve dependencies
among columns. A matroid isomorphic to such an $M[G]$ is said to be
\emph{$\F_q$-representable}.

\begin{definition}
\label{def:closure}
For a matroid $M=(E,r)$, a set $A\subseteq E$ is \emph{independent} if $r(A)=|A|$ and \emph{dependent} otherwise. A minimal dependent set is called a \emph{circuit}. Moreover, the \emph{closure} of $A$ is
\[
  \cl_M(A)=\{e\in E:r(A\cup\{e\})=r(A)\}.
\]

\end{definition}

For a representable matroid, adjoining an index to $A$ leaves its rank unchanged precisely when the corresponding column lies in the
span of the columns indexed by $A$. Hence
\begin{equation}\label{eq:closure-span}
  i\in\cl_{M[G]}(A)
  \iff g_i\in\langle g_j:j\in A\rangle.
\end{equation}
This allows us to express $\aalpha_i(s)$ in terms of matroid closure.

\begin{definition}
\label{def:loop}
An element $e\in E$ is a \emph{loop} of $M$ if $r(\{e\})=0$,
and a \emph{coloop} if $r(E\setminus\{e\})=r(E)-1$.
Equivalently, a coloop belongs to every maximal independent set.
\end{definition}

In $M[G]$, an index $i$ is a loop if and only if $g_i=0$, and a
coloop if and only if $g_i\notin\langle g_j:j\ne i\rangle$.
Since $r(\varnothing)=0$, the closure definition gives 
\begin{itemize}
\item  $i\in\cl_M(\varnothing)$ if and only if $i$ is a loop of $M$.
\item $ i\in\cl_M(E\setminus\{i\})$ if and only if $i$ is not a coloop of $M$.
\end{itemize}

\begin{definition}
\label{def:dual}
The \emph{dual} of $M=(E,r)$ is the matroid $M^*=(E,r^*)$, where
\begin{equation}\label{eq:dual-rank}
  r^*(A)=|A|+r(E\setminus A)-r(E),\qquad A\subseteq E.
\end{equation}
The function $r^*$ satisfies \textnormal{(R1)--(R3)}, and
$(M^*)^*=M$.
\end{definition}

We will use the following standard consequences of duality, which follow directly from the definitions.

\begin{proposition}[\cite{oxley2006matroid}]\label{prop:dual-facts}
\begin{enumerate}
  \item[\textnormal{(a)}] An element is a loop of $M^*$ if and
  only if it is a coloop of $M$, and a coloop of $M^*$ if and
  only if it is a loop of $M$
  \textnormal{\cite[\S2.1]{oxley2006matroid}}.
  \item[\textnormal{(b)}] If $G$ and $H$ are generator and
  parity-check matrices of $\code$, respectively, then
  $M[H]=(M[G])^*$, or equivalently
  $M(\duale{\code})=M(\code)^*$
  \textnormal{\cite[\S2.2, Thm.~2.2.8]{oxley2006matroid}}.
\end{enumerate}
\end{proposition}

\begin{comment}
\begin{IEEEproof}
For (a), \eqref{eq:dual-rank} gives
\[
  r^*(\{e\})=1+r(E\setminus\{e\})-r(E),
\]
which is zero exactly when $e$ is a coloop of $M$.
The reverse correspondence follows from $(M^*)^*=M$.
Part (b) is the standard duality relation between the vector
matroids of a generator matrix and a parity-check matrix; see
\cite[\S2.2]{oxley2006matroid} or
\cite[Ch.~7]{macwilliams1977theory}.
\end{IEEEproof}
\end{comment}

The next lemma relates closure in $M$ to closure in $M^*$ and
will be used to prove the main duality identity.
We give its proof for completeness. 
\begin{lemma}[\cite{oxley2006matroid}]\label{lem:complement}
Let $M$ be a matroid on $E=[n]$, let $i\in[n]$, and let
$S\subseteq[n]\setminus\{i\}$. Set
$T=[n]\setminus(S\cup\{i\})$. Then
\[
  i\in\cl_M(S)\iff i\notin\cl_{M^*}(T).
\]
\end{lemma}

\begin{IEEEproof}
Since $E\setminus T=S\cup\{i\}$ and
$E\setminus(T\cup\{i\})=S$, \eqref{eq:dual-rank} gives
\[
\begin{aligned}
  r^*(T)&=|T|+r(S\cup\{i\})-r(E),\\
  r^*(T\cup\{i\})&=|T|+1+r(S)-r(E).
\end{aligned}
\]
Subtracting yields
\[
  r^*(T\cup\{i\})-r^*(T)
  =1-\bigl(r(S\cup\{i\})-r(S)\bigr).
\]
By \textnormal{(R1)--(R3)}, the rank increment
$r(S\cup\{i\})-r(S)$ belongs to $\{0,1\}$.
It is zero exactly when $i\in\cl_M(S)$, in which case the dual
increment is one and $i\notin\cl_{M^*}(T)$; otherwise, the dual
increment is zero and $i\in\cl_{M^*}(T)$.
\end{IEEEproof}
%==========================================

%===============================================================
\section{The expectation as a matroid invariant}\label{sec:reform}
%===============================================================

Throughout this section $\code$ is an $[n,k]_q$ code with generator matrix $G$ and matroid $M=M(\code)$. By \eqref{eq:closure-span}, the counting function $\aalpha_i$ of \eqref{eq:expectation} is purely matroidal:
\begin{equation}\label{eq:alpha-matroid}
  \aalpha_i(s) \;=\; \bigl|\{ S\subseteq[n] : |S|=s,\ i \in \cl_M(S) \}\bigr| .
\end{equation}
In particular $\aalpha_i(s)$, and therefore (by (\ref{eq:expectation})) $\E[\ttau_i(\code)]$, depends only on the matroid $M(\code)$; this recovers and refines \cite[Claim~2]{gruica2025combinatorial}.

It is convenient to separate, in \eqref{eq:alpha-matroid}, the subsets that contain~$i$ from those that do not.

\begin{definition}\label{def:a}
For $i\in[n]$ and $0\le s\le n-1$, set
\[
  a_i^M(s) \;:=\; \bigl|\{ S\subseteq[n]\setminus\{i\} : |S|=s,\ i\in\cl_M(S) \}\bigr|,
\]
and
\[
  A_i^M := \sum_{s=1}^{n-1} \frac{a_i^M(s)}{\binom{n-1}{s}}.
  \]
\end{definition}

\begin{lemma}\label{lem:reduced}
Let $i\in[n]$ be such that $g_i \ne 0$ (i.e.\ $i$ is not a loop of $M$). Then
\[
  \aalpha_i(s) = \binom{n-1}{s-1} + a_i^M(s) \quad(1\le s\le n-1),
\]
and $\E[\ttau_i(\code)] \;=\; n - A_i^M$.
\end{lemma}

\begin{IEEEproof}
If $i\in S$ then, it holds that $i\in\cl_M(S)$, and there are $\binom{n-1}{s-1}$
such subsets of size $s$. Otherwise, $i\notin S$ and the subsets not containing $i$ are counted by $a_i^M(s)$. This proves the first identity. Substituting it into \eqref{eq:expectation} gives
\[
  \E[\ttau_i(\code)] = nH_n - \sum_{s=1}^{n-1}\frac{\binom{n-1}{s-1}}{\binom{n-1}{s}} - A_i^M .
\]
Since $\binom{n-1}{s-1}/\binom{n-1}{s} = s/(n-s)$, the middle sum equals
\[
  \sum_{s=1}^{n-1}\frac{s}{n-s}
  \;=\; \sum_{t=1}^{n-1}\frac{n-t}{t}
  \;=\; n H_{n-1} - (n-1) \;=\; nH_n - n,
\]
where the first equality follows from substituting $t=n-s$.
Therefore $\E[\ttau_i(\code)] = nH_n - (nH_n - n) - A_i^M = n - A_i^M$.
\end{IEEEproof}

Lemma~\ref{lem:reduced} already explains the two extreme cases. If $i$ is a coloop,
then no subset $S\not\ni i$ has $i\in\cl_M(S)$, so $a_i^M(s)=0$ for all $s$, whence
$A_i^M=0$ and $\E[\ttau_i(\code)]=n$. At the other extreme, for an MDS code one
recovers $\E[\ttau_i(\code)]=k$ (cf.\ \cite[Proposition~1]{gruica2025combinatorial}) since $a_i^M(s)=0$ for $1 \leq s \leq k-1$ and $a_i^M(s)=\binom{n-1}{s}$ for $k\leq s \leq n-1$.

As we already observed, the property of being recovery balanced is determined by the associated matroid, and actually at the matroidal level,  there is no need of representability. Indeed, the notion of recovery balance can be naturally extended to arbitrary matroids as follows. 
\begin{definition}
    A matroid $M=([n],r)$ is said to be recovery balanced if $A_i^M=A_j^{M}$ for every $i,j\in[n]$.
\end{definition}
It is clear from Lemma~\ref{lem:reduced} that a code is recovery balanced if and only if the associated matroid is recovery balanced.
We conclude the section with 
the following identity, which is the main tool for our duality result.

\begin{lemma}\label{lem:reflection}
Let $M$ be a matroid on $[n]$ and $i\in[n]$. For all $0\le s\le n-1$,
\[
  a_i^M(s) \;=\; \binom{n-1}{s} - a_i^{M^*}(n-1-s).
\]
\end{lemma}

\begin{IEEEproof}
The map $S \mapsto T := ([n]\setminus\{i\})\setminus S$ is a bijection from the
$s$-subsets of $[n]\setminus\{i\}$ onto its $(n-1-s)$-subsets. By
Lemma~\ref{lem:complement}, $i\in\cl_M(S)$ if and only if $i\notin\cl_{M^*}(T)$.
Hence,
\begin{equation*}
\begin{split}
  a_i^M(s) &= \binom{n-1}{\,n-1-s\,} - a_i^{M^*}(n-1-s)
            \\ &= \binom{n-1}{s} - a_i^{M^*}(n-1-s), 
\end{split}
\end{equation*}
since the number of $s$-subsets $S\subseteq[n]\setminus\{i\}$ with $i\in\cl_M(S)$
equals the number of $(n-1-s)$-subsets $T\subseteq[n]\setminus\{i\}$ with
$i\notin\cl_{M^*}(T)$.
\end{IEEEproof}

%===============================================================
\section{The duality identity and proof of Conjecture~\ref{conj:main}}\label{sec:main}
%===============================================================

Recall that a code $\code\subseteq\F_q^n$ is \emph{non-degenerate} if no coordinate
is identically zero on $\code$, i.e.\ if $G$ has no zero column. Note that if both $\code$ and $\duale{\code}$ are non-degenerate, then $M(\code)$ has neither
loops nor coloops. All codes considered in \cite{gruica2025combinatorial} (MDS, Hamming,
simplex, Reed--Muller, Golay codes) are of this type.

\begin{theorem}\label{thm:identity}
Let $M=([n],r)$ be a matroid. Then for every $i\in[n]$,
\[
  A_i^{M} + A_i^{M^*} = n-\ell_i - c_i,
\]
where $\ell_i = 1$ if $i$ is a loop and $c_i = 1$ if $i$ is a coloop of $M$.
\end{theorem}
\begin{IEEEproof}
Using Lemma~\ref{lem:reflection}, 
\begin{align*}
  A_i^{M}
  &= \sum_{s=1}^{n-1}\frac{a_i^M(s)}{\binom{n-1}{s}}
  = \sum_{s=1}^{n-1}\left(1 - \frac{a_i^{M^*}(n-1-s)}{\binom{n-1}{\,n-1-s\,}}\right)\\
  &= (n-1) - \sum_{s'=0}^{n-2}\frac{a_i^{M^*}(s')}{\binom{n-1}{s'}},
\end{align*}
where we re-indexed $s' = n-1-s$. By isolating the terms $s'=0$ and
$s'=n-1$, we get
\[
  \sum_{s'=0}^{n-2}\frac{a_i^{M^*}(s')}{\binom{n-1}{s'}}
  = a_i^{M^*}(0) + A_i^{M^*} - a_i^{M^*}(n-1).
\]
By Definition~\ref{def:a}, we have
\begin{equation*}
    a_i^{M^*}(0)=\begin{cases}1&\text{ if }i\text{ is a loop in }M^*\\
    0&\text{otherwise,}
    \end{cases}
\end{equation*}
and
\begin{equation*}
    a_i^{M^*}(n-1)=\begin{cases}0&\text{ if }i\text{ is a coloop in }M^*\\
    1&\text{otherwise.}
    \end{cases}
\end{equation*}
Substituting,
\[
  A_i^{M} = n-A_i^{M^*}-a_i^{M^*}(0)-(1-a_i^{M^*}(n-1)),
\]
that is $A_i^{M}+A_i^{M^*}=n-\ell_i-c_i$, as required.
\end{IEEEproof}

We can now settle the conjecture. We first exclude the degenerate case in the following lemma. The proof is immediate from the definitions of loops, coloops and the fact that for a recovery balanced code, $\E[\ttau_i(\code)]=k$ for all $i\in [n]$.

\begin{lemma}\label{lem:nondeg}
Let $\code$ be a recovery balanced $[n,k]_q$ code with $0<k<n$. Then $M(\code)$ has
neither loops nor coloops.
\end{lemma}
\begin{comment}
\begin{IEEEproof}
By Equation~\eqref{eq:average} and Definition~\ref{def:rb}, the common value of the
$\E[\ttau_i(\code)]$ is $k$. Suppose $i$ were a coloop of $M(\code)$. As observed
after Lemma~\ref{lem:reduced}, then $\E[\ttau_i(\code)] = n$, forcing $k=n$, a
contradiction. Suppose instead $i$ were a loop, i.e.\ $g_i=0$. Then $g_i$ lies in
the span of the empty set, so $\ttau_i(\code)=0$ deterministically and
$\E[\ttau_i(\code)] = 0$, forcing $k=0$, again a contradiction. Hence $M(\code)$ has
no loops and no coloops.
\end{IEEEproof}
\end{comment}
We are now in position to prove Conjecture~\ref{conj:main}.
\begin{theorem}\label{thm:main}
A code $\code$ is recovery balanced if and only if its dual $\duale{\code}$ is
recovery balanced~\footnote{An AI-generated proof for the case of linear codes, prepared with OpenAI Codex (GPT-6 Astra), was posted pseudonymously online in October 2026. Its accompanying disclosure states that no independent human audit has been conducted on the formal Lean statement. The proof presented here was obtained independently.}.
\end{theorem}
\begin{IEEEproof}
Since $(\duale{\code})^{\perp}=\code$, it suffices to show that if $\code$ is recovery balanced then so is $\duale{\code}$. Let $\code$
be a recovery balanced $[n,k]_q$ code.
If $k=0$, then $\code=\{0\}$ and $\duale{\code}=\F_q^n$, while if $k=n$, then
$\code=\F_q^n$ and $\duale{\code}=\{0\}$. It can be verified that in either case both codes are recovery balanced.

Now assume $0<k<n$. By Lemma~\ref{lem:nondeg}, $M(\code)$ has neither loops nor
coloops, so Theorem~\ref{thm:identity} applies and yields, for every $i\in[n]$,
\[
  \E[\ttau_i(\duale{\code})] = n - \E[\ttau_i(\code)] = n - k,
\]
where we used that the common balanced value equals $k$. As this value is
independent of $i$, the dual $\duale{\code}$ is recovery balanced (with common value
$n-k$, consistently with \eqref{eq:average} for $\duale{\code}$).
\end{IEEEproof}
\begin{remark}\label{rem:pointwise}
Theorem~\ref{thm:identity} is strictly stronger than what
Conjecture~\ref{conj:main} requires. Indeed, it holds coordinatewise, with no assumption of balance. The entire ``profile'' $\bigl(\E[\ttau_i(\duale{\code})]\bigr)_{i\in[n]}$
of the dual is the reflection $\E[\ttau_i] \mapsto n - \E[\ttau_i]$ of that of $\code$ when both $\code$ and $\duale{\code}$ are non-degenerate.
\end{remark}

\section{Sparse paving matroids}\label{sec:sparse}
Over a sufficiently large finite field, a random $[n,k]_q$ code is MDS \cite{GhorpadeLachaud2001} and, therefore, recovery balanced. However, notice that every $[n,k]_q$ MDS code is associated with the same matroid, known as the uniform matroid. So, a natural question to ask is whether a random matroid is recovery balanced. To this end, we study the proportion of sparse paving matroids that are not recovery balanced as $n$ approaches infinity. We begin by reviewing the necessary definitions.
\begin{definition}\label{def:sp}
A matroid $M$ of rank $k$ on $[n]$ is paving if and only if every circuit has cardinality at least $k$. 
A matroid $M$ of rank $k$ on $[n]$ is sparse paving if and only if it is paving and its dual is paving. 
\end{definition}
For a sparse paving matroid it is known that
\begin{enumerate}
  \item[(a)] every subset of $[n]$ of cardinality at most $k-1$ is independent;
  \item[(b)] any two distinct dependent sets of cardinality $k$ intersect in at most $k-2$ elements.
\end{enumerate}
\begin{comment}
Item (a) is a direct consequence of the fact that $M$ is paving. As for item (b), suppose there are two circuits $C_1$ and $C_2$ of $M$ of cardinality $k$ such that their intersection has cardinality $k-1$. The set $U=C_1\cup C_2$ does not contain any basis. In fact, every basis of $M$ has cardinality $k$, and therefore if there were a basis in $U$, it would be of the form $U\setminus \{i\}$ with $i\in U$. However, by the circuit axioms, $U\setminus \{i\}$ contains a circuit and therefore cannot be a basis. It follows that $[n]\setminus U$ cannot be independent in $M^*$, otherwise $U$ would contain a basis for $M$. So $[n]\setminus U$ is dependent, but since $\lvert [n]\setminus u\rvert=n-k-1$, this is in contradiction with the fact that $M^*$ is paving.
\end{comment}
We will use the fact that sparse paving matroids are in bijection with the set of independent sets of the Johnson graph. 
\begin{definition}
The Johnson graph $J(n,k)$ is the graph whose vertex set is $\binom{[n]}{k}$, the set of subsets of $[n]$ with cardinality $k$, and two vertices are adjacent if and only if their intersection has cardinality $k-1$. A family
$\mathcal{A}\subseteq\binom{[n]}{k}$ is an independent set of $J(n,k)$ if no two of its elements are adjacent. We denote by $\I(n,k)$ the collection of independent sets of $J(n,k)$ and we set $I(n,k)=|\I(n,k)|$.
\end{definition}
We write $\dH(M)$ for the family of dependent $k$-sets of $M$ and let $m=\lvert\dH(M)\rvert$.
\begin{lemma}[{\cite[Lemma 6]{bansal2015number}}]\label{lemma:bijection}
Let $1\le k\le n-1$. The map $M\mapsto\dH(M)$ is a bijection between the sparse paving matroids of rank $k$ on $[n]$ and the independent sets of $J(n,k)$. In particular, the number of sparse paving matroids of rank $k$ on $[n]$ equals $I(n,k)$.
\end{lemma}
For all $i\in [n]$, consider
\begin{equation*}
    d_i(M)=\lvert\{H\in\dH(M) : i\in H\}\rvert.
\end{equation*}
Let us now express the $a_i^M(s)$ of a sparse paving matroid in terms of the $d_i(M)$.
\begin{lemma}\label{lemma:alpha}
Let $M$ be a sparse paving matroid of rank $k$ on $[n]$ with $2\leq k\leq n-2$. Then, for every $i\in[n]$,
\begin{equation*}
   a_i^M(s)=
  \begin{cases}
    0 &\text{ if } 1\leq s\leq k-2,\\
    d_i(M) &\text{ if } s=k-1,\\
    \binom{n-1}{k}-m+d_i(M) &\text{ if } s=k,\\
    \binom{n-1}{s} &\text{ if } k+1\le s\le n-1.
  \end{cases} 
\end{equation*}
\end{lemma}
\begin{comment}
\begin{IEEEproof}
Fix $i$ and let $S\subseteq[n]$ with $i\notin S$ and $|S|=s$. If $s\leq k-2$, since every subset of cardinality $k-1$ is independent, then $\cl(S)=S\not\ni i$, and so $\alpha_i^M(s)=0$.
If $s=k-1$, then $i\in\cl(S)$ if and only if
$S\cup\{i\}\in\dH(M)$, i.e.\ $S=C\setminus\{i\}$ for some $C\in\dH(M)$ with $i\in C$. So, $\alpha_i^M(k-1)=d_i(M)$.
If $s=k$, then $i\in\cl(S)$ if and only if $S\notin\dH(M)$, in which case $\cl(S)=[n]$. This follows quite directly from the fact that the intersection among two circuits of cardinality $k$ has intersection of size at most $k-2$. The number of $k$-sets avoiding $i$ is
$\binom{n-1}{k}$, of which exactly $m-d_i(M)$ belong to $\dH(M)$. Hence,
$\alpha_i^M(k)=\binom{n-1}{r}-(m-d_i(M))$. Finally, if $s\ge k+1$, then $\cl(S)=[n]$ for every subset $S$ of cardinality $s$,, so $\alpha_i^M(s)=\binom{n-1}{s}$.
\end{IEEEproof}
\end{comment}
From Lemma~\ref{lemma:alpha}, it is immediately clear that the condition for a matroid to be recovery balanced can be expressed in terms of the $d_i(M)$'s.
\begin{corollary}\label{cor:regular}
Let $M$ be a sparse paving matroid of rank $k$ on~$[n]$ with $2\le k\le n-2$. Then
$M$ is recovery balanced if and only if
$d_1(M)=d_2(M)=\dots=d_n(M)$. In that case $n\mid km$ and $d_i(M)=km/n$ for all $i \in [n]$.
\end{corollary}
\begin{comment}
\begin{IEEEproof}
    The first part of the statement follows immediately from Lemma~\ref{lemma:alpha}. For the second part, note that every $i$ belong to the same number of dependent $k$-sets.
\end{IEEEproof}
\end{comment}
We say that an independent set $\dH\in\I(n,k)$ is regular if $d_1(\dH)=\dots=d_n(\dH)$. In view of Corollary~\ref{cor:regular}, counting recovery balanced sparse paving matroids of rank $k$ amounts to counting regular independent sets of the Johnson graph $J(n,k)$.
Our goal now is to determine the proportion of regular independent sets in a Johnson graph. We do this by estimating separately the cardinality of the families of regular independent sets with small and large cardinalities; proofs are deferred to the long version in the interest of space.
\begin{lemma}\label{lemma:bipartitegraph}
Let $1\le k\le n-1$ and $t\ge 1$. Then
\begin{equation*}
    \lvert \{\dH\in\I(n,k):\dH \text{ is regular and }|\dH|\geq t\}\rvert
   \leq\frac{I(n,k)}{t+1}.
\end{equation*}
\end{lemma}
\begin{comment}
\begin{IEEEproof}
We consider a bipartite graph $G$ whose set of vertices is $\I(n,r)$ as follows. The edges of $G$ are given by the pairs $(\dH_1,\dH_2)$ where $\dH_1$ is regular and $\dH_2$ is obtained by $\dH_1$ by removing an element. The graph $G$ is clearly bipartite between the regular independent sets and the non-regular ones. We observe that $\dH_2$ cannot be regular, because we are removing only one element so exactly $k$ of the $d_i(\dH_1)$ decrease by one, while $n-k$ remain constant. Moreover, this also proves that each non-regular set is connected to at most one regular set.  

Counting the edges of $G$ from both side of $G$, we get
\begin{equation*}
    t\lvert \{\dH\in\I(n,k):\dH \text{ is regular and }|\dH|\geq t\}\rvert\leq I(n,k)-\lvert\{\dH\in\I(n,k):\dH \text{ is regular and }|\dH|\geq t\}\}\rvert
\end{equation*}
which is exactly the claim.
\end{IEEEproof}
\end{comment}
\begin{lemma}\label{lemma:extrcases}
    Let $1\le k\le n-1$. Then
    \begin{equation*}
        \lvert \{\dH\in\I(n,k):\dH \text{ is regular}\}\rvert
   \leq\frac{\min\{k,n-k\}}{n}I(n,k).
    \end{equation*}
\end{lemma}
\begin{comment}
\begin{IEEEproof}
    Let $G$ be as in the proof of Lemma~\ref{lemma:bipartitegraph}. Then, by counting the edges of $G$ on both side we get
    \begin{equation*}
        \sum_{\dH\in\I(n,k)\setminus\emptyset;\\
        \dH\text{ is regular}}\lvert\dH\rvert\leq I(n,k).
    \end{equation*}
    By Corollary~\ref{cor:regular} we get $\lvert\dH \rvert\geq n/k$. Let $M$ be the sparse paving matroid corresponding to $\dH$. Since $M$ is sparse paving, the circuits of size $n-k$ in the dual matroid of $M$ are exactly the complements of the circuits of size $k$ in $M$. Therefore, by duality we get $\lvert\dH \rvert\geq n/(n-k)$. 
\end{IEEEproof}
\end{comment}
\begin{lemma}\label{lem:small}
Let $1\le k\le n-1$ and $t\ge 1$, and suppose $n\geq 2$. Then
\begin{equation*}
    \lvert \{\dH\in\I(n,k):\dH \text{ is regular and }|\dH|< t\}\rvert\leq 2^{nt}
\end{equation*}
\end{lemma}
\begin{IEEEproof}
    Clearly the number of regular independent sets of size at most $t$ is bounded by the number of independent sets of size at most $t$. Moreover,
    \begin{equation*}
        \lvert \{\dH\in\I(n,k):|\dH|<t\}\rvert\leq\sum_{i=0}^{t-1}\binom{N_k}{i}\leq \sum_{i=0}^{t-1}N_k^i\leq 2^{nt},
    \end{equation*}
    where $N_k=\binom{n}{k}$.
\end{IEEEproof}
To determine a lower bound on the number of sparse paving matroids, inspired by~\cite[Equation 2]{mayhew2013number}, we use~\cite[Theorem 1]{graham1980lower} from Graham and
Sloane.
\begin{theorem}\label{theorem:grahamsloane}
Let $A(n,4,w)$ denote the maximum number of codewords in any binary code of length $n$, minimum distance~$4$, and constant weight $w$. Then,
\begin{equation*}
    A(n,4,k)\geq\frac{1}{n}\binom{n}{k}.
\end{equation*}
\end{theorem}
\begin{lemma}\label{lemma:sparse}
Let $N_k=\binom{n}{k}$. Then $I(n,k)\geq 2^{N_k/n}$.
\end{lemma}
\begin{IEEEproof}
    By Theorem~\ref{theorem:grahamsloane}, there exists a code $\mathcal{C}$ of minimum distance $4$ and constant weight $k$ of cardinality at least $\frac{1}{n}\binom{n}{k}$. Any subset of $\mathcal{C}$ corresponds to an independent set of  $J(n,k)$. We conclude that $\mathcal I(n,k)$ has cardinality at least $2^{\lvert\mathcal{C}\rvert}$.
\end{IEEEproof}
We can now prove the main theorem of this section.
\begin{theorem}\label{theorem:4/n}
    There exists $n_0\in\mathbb{N}$ such that for every $n\geq n_0$ and every $1\leq k\leq n-1$ we have
    \begin{equation*}
       \frac{\lvert \{\dH\in\I(n,k):\dH \text{ is regular}\}\rvert}{I(n,k)}\leq\frac{4}{n}
    \end{equation*}
\end{theorem}
\begin{IEEEproof}
    If $k\leq 4$ or $k\geq n-4$, the result follows directly from Lemma~\ref{lemma:extrcases}. So, assume that $4<k<n-4$. In particular, when $n$ is large enough we get $N_k=\binom{n}{k}\geq n^4$. We set $t=\lceil N_k/n^3\rceil$. By Lemmas~\ref{lemma:bipartitegraph} and~\ref{lem:small},
    \begin{equation*}
    \begin{split}
        \lvert \{\dH\in\I(n,k):\dH \text{ is regular}\}\rvert &\leq \frac{I(n,k)}{\lceil N_k/n^3\rceil}+2^{nt}\\ &\leq\frac{I(n,k)}{n}+2^{nt}.
    \end{split}
    \end{equation*}
    We divide both sides by $I(n,k)$ and by Lemma~\ref{lemma:sparse} for $n$ large enough we get
    \begin{equation*}
      \frac{\lvert \{\dH\in\I(n,k):\dH \text{ is regular}\}\rvert}{I(n,k)}\leq \frac{1}{n}+2^{\frac{n^2\lceil N_k/n^3\rceil-N_k}{n}}\leq\frac{4}{n}.  
    \end{equation*} 
\end{IEEEproof}
The previous theorem can also be restated as follows.
\begin{corollary}\label{coro:recbalsparse}
    Asymptotically almost every sparse paving matroid is not recovery balanced, that is
    \begin{equation*}
        \lim_{n\to\infty}\frac{\lvert\{\mathrm{\text{rec. balanced sparse paving matroids on}}\,[n]\}\rvert}{\lvert\{\mathrm{\text{sparse paving matroids on}}\,[n]\}\rvert}=0.
    \end{equation*}
\end{corollary}
We believe that the assumption regarding sparse paving in the previous theorem can be eliminated. That is, we believe that asymptotically almost every matroid is not recovery balanced.

\begin{remark}
We point out that in~\cite[Conjecture 1.6]{mayhew2011asymptotic} Mayhew, Newman, Welsh, and Whittle conjectured that almost every matroid is sparse paving. This is strongly supported by the counting results of Bansal, Pendavingh and van der Pol \cite{bansal2015number}. A positive answer to their conjecture, together with Theorem~\ref{theorem:4/n}, would imply that indeed  almost every matroid is not recovery balanced when $n\to\infty$.
 
\end{remark}
\begin{remark}
    Note that Corollary~\ref{coro:recbalsparse} concerns the uniform measure of all sparse paving matroids and implies nothing about representable ones. Recall that a linear $[n,k]$ code $C$ is near-MDS when $d(C)=n-k$ and $d(C^\perp)=k$; since $d(C)$ and $d(C^\perp)$ are the minimum cocircuit and circuit sizes of the matroid $M(C)$, the matroids of near-MDS codes are exactly the non-uniform representable sparse paving matroids. However, the set of representable matroids asymptotically has measure zero in the set of matroids, and so our approach gives zero information on the size of representable sparse paving matroids.
\end{remark}

\section*{Acknowledgment}
This work was initiated during the “Coding Theory and Cryptography Summer School and Collaboration Workshop”, held at the Stager Center for International Scholarship – Virginia Tech, Switzerland, in July 2026. The authors thank the organizers for facilitating this collaboration.

\section*{AI Disclosure}
Claude Fable suggested proving Theorem~\ref{theorem:4/n} by distinguishing between small and large cardinalities, and suggested a connection of our proof of Theorem~\ref{thm:main} to matroid theory, which was further developed by the authors. ChatGPT-6 Astra was used for grammar and wording.
\bibliographystyle{ieeetr}
\bibliography{bib}

\end{document}